\documentclass[12pt]{article}

\usepackage{algorithm,algpseudocode}
\usepackage{amscd,amsfonts,amsopn,amssymb,amstext}
\usepackage{appendix}
\usepackage{booktabs}
\usepackage{bbm,bm}
\usepackage[centertags]{amsmath}
\usepackage{color}
\usepackage{fullpage}
\usepackage{graphicx,graphics,psfrag}
\usepackage{indentfirst}
\usepackage{latexsym,enumerate}
\usepackage{lscape}
\usepackage{multirow}
\usepackage{natbib}
\usepackage{rotating}
\usepackage{setspace}
\usepackage{srcltx}
\usepackage{subfigure}
\usepackage[T1]{fontenc}
\usepackage{threeparttable}
\usepackage{times}
\usepackage{url}
\usepackage{verbatim}

\makeatletter
\renewcommand\@biblabel[1]{#1.}
\makeatother

\newtheorem{proposition}{\vspace{-0.05cm}\newline{Proposition}}

\newtheorem{example}{\vspace{-0.2cm}\newline{Example}}

\newtheorem{definition}{Definition}

\newenvironment{proof}{{\bf Proof:\ \ }}{\qed}
\newcommand{\qed}{\rule{0.5em}{1.5ex}}

\title{\Large \bf{Modeling citation concentration through a mixture of Leimkuhler curves}}

\author{\normalsize
{\bf Emilio G\'omez--D\'eniz, Pablo Dorta--Gonz\'alez\thanks{P. Dorta--Gonz\'alez, Department of Quantitative Methods in Economics,
University of Las Palmas de Gran Canaria, 35017 Las Palmas de Gran Canaria, Spain. E-mail: \texttt{pablo.dorta@ulpgc.es}} },\\
{\small Department of Quantitative Methods in Economics and TIDES Institute,}\\[-0.2cm]
{\small University of Las Palmas de Gran Canaria, Spain}\\[-0.15cm]
{\small http://orcid.org/0000-0002-5072-7908 (EGD) }\\[-0.20cm]
{\small http://orcid.org/0000-0003-0494-2903(PDG)}
}

\date{}

\def \E{{\rm I\kern -2.2pt  E}}

\begin{document}

\maketitle

\vspace{-0.5cm}

\begin{abstract}\noindent
When a graphical representation of the cumulative percentage of total citations to articles, ordered from most cited to least cited, is plotted against the cumulative percentage of articles, we obtain a Leimkuhler curve. In this study, we noticed that standard Leimkuhler functions may not be sufficient to provide accurate fits to various empirical informetrics data. Therefore, we introduce a new approach to Leimkuhler curves by fitting a known probability density function to the initial Leimkuhler curve, taking into account the presence of a heterogeneity factor. As a significant contribution to the existing literature, we introduce a pair of mixture distributions (called PG and PIG) to bibliometrics. In addition, we present closed-form expressions for Leimkuhler curves. {Some measures of citation concentration are examined empirically for the basic models (based on the Power {and Pareto distributions}) and the mixed models derived from {these}.} An application to two sources of informetric data was conducted to see how the mixing models outperform the standard basic models. The different models were fitted using non-linear least squares estimation.

\vspace{-0.5cm}

\end{abstract}

\paragraph{Funding:} Emilio G\'omez-D\'eniz was partially funded by grant PID2021-127989OB-I00 (Ministerio de
Econom\'ia y Competitividad, Spain).

{\paragraph{Acknowledgements}
We thank the two anonymous reviewers and the Associated Editor for their valuable comments and suggestions, which have greatly helped us improve the original manuscript.}

\paragraph{Highlights: }
\begin{itemize}
\setlength\itemsep{0.001em}
\item	Introduction of mixture distributions (called PG and PIG) to bibliometrics.
\item	Closed-form expressions for Leimkuhler curves.
\item	Calculations with mixture distributions.
\item	Theoretical and mathematical results in bibliometrics/informetrics.
\end{itemize}

\section{Introduction}

Power law relationships or similar long-tailed distributions are pervasive in various domains of informetrics, such as authors and their publications (\citealp{lotka_1926}), topics and journals related to them \citetext{\citealp{bradford_1934}}, web sites and their links \citetext{\citealp{faloutsos_1999}}, or research literature and its obsolescence (\citealp{dortaandgomez_2022}). Furthermore, these types of distributions are widely observed in diverse fields, including demography (cities and their inhabitants), linguistics (words and their uses), economics (incomes in a market economy), ecology (forest fragmentation), astronomy (initial mass functions), and more \citetext{\citealp{newman_2005}; \citealp{pareto_1895}; \citealp{salpeter_1955}; \citealp{saravia_2018}; \citealp{zipf_1941}}. The ubiquity of power laws is evident across different domains, showcasing their widespread presence in various disciplines.

Power laws are mathematical relationships that describe a phenomenon in which a small number of elements have a disproportionately high impact or frequency compared to the majority. In the context of citation networks, power laws suggest that a few highly cited papers have a much greater impact than the majority of less cited papers. This phenomenon is known as preferential attachment or cumulative advantage (\citealp{brzezinski_2015}). Some other papers on the application of power laws to citations are those of \cite{thelwallandwilson_2014}, \cite{thelwall_2016a}, and \cite{thelwall_2016b}. The distribution of citations among scientific papers often follows a power law distribution, which is characterized by a straight line on a log-log plot. This means that while most papers receive few citations, a small fraction receive an exceptionally large number of citations and contribute significantly to the overall impact of the field. This pattern is similar to the well-known "80/20" rule, where roughly 80\% of citations are contributed by 20\% of papers.

Lorenz and Leimkuhler curves have proven to be useful tools in bibliometrics for analyzing the distribution and growth of scientific output. These curves can provide insights into the concentration of research output, the distribution of citations, and the factors influencing the growth rate of scientific output. For example, a study by \cite{abramoetal_2011} used Lorenz curves to analyze the concentration of research output among Italian universities. {\cite{egghe_2006}} examined the distribution of citations among scientific papers. He found that the distribution was highly skewed, with a few papers receiving a disproportionate number of citations. More recently, \cite{nairandvineshkumar_2022} derived novel properties of Leimkuhler curves by employing quantile functions and applied these findings to the informetrics domain.

These distributions share the common characteristic of having a few sources with a high concentration of items, leading to the study of inequality or concentration as a major focus of the field (\citealp{rousseau_2018}). The social implications of this are significant, prompting research into techniques for quantifying and modeling inequality. One such technique is Leimkuhler ordering, which provides a partial ordering of non-negative random variables based on their concentration properties. This allows the comparison of distributions and the identification of more concentrated ones.

An area of interest within the field of informetrics, particularly in bibliometrics, revolves around the challenge of measuring journal productivity through citation analysis to identify a core of influential journals that are at the forefront of research in a particular field. This issue concerns the distribution pattern of citations among different journals over a given period of time and how to rank journals based on the decreasing number of citations received. Research on this topic has received considerable attention in recent decades. For a detailed study of this problem, see for example \cite{hubert_1977}. If a graphical representation of the cumulative percentage of the total sample of cited papers, ordered from largest to smallest, is plotted against the cumulative percentage of journals, we get a Leimkuhler curve, or double cumulative curve (the curve starts at the largest value) of a productivity distribution, as it is known in the informetric literature. It is well known that this tool is related to the graph of the Lorenz curve used in economics to model income distribution \citetext{see \citealp{burrell_2005}; {\citealp{sarabia_2008a,sarabia_2008b}}; and \citealp{sarabiaetal_2010}, among others}.

In this research, our findings suggest that conventional Leimkuhler functions may not adequately capture the complexities of diverse empirical informetrics data. To address this limitation, we propose a novel approach to Leimkuhler curves that involves fitting a known probability density function to the original Leimkuhler curve, while accounting for the presence of a heterogeneity factor. We then analyze various citation concentration measurements using both the traditional basic models {(based on Power and Pareto distributions)} and the mixed models derived from them. To assess the effectiveness of our approach, we apply it to two sources of informetric data and compare the performance of the mixed models with the standard basic models.

The structure of this paper is as follows. Section \ref{s2} deals with the standard definition of a Leimkhuler curve and provides two standard Leimkhuler curves, the one based {on the power distribution and the one based on the Pareto distribution}. Two numerical examples are provided to show how these two basic models work. The mixture model of Leimkhuler curves is shown in section \ref{s3}. Here we also give the expression for the Gini and Pietra indices, two measures of citation concentration. Numerical applications are given in section \ref{s4} and conclusions in the last section \ref{s5}.

\section{The standard Leimkuhler curve\label{s2}}
The Leimkuhler curve is a crucial instrument in information processing, especially in informetrics \citetext{see for instance, \citealp{burrell_2005}}. This curve is closely associated with the Lorenz curve, differing only in that the contributions are arranged in a descending order along the Leimkuhler curve. This curve plots the cumulative proportion of total citations against the cumulative proportion of sources. Specifically, given a set of $n$ ordered numbers $x_1\geq x_2\geq\dots\geq x_n\geq 0$, the Leimkuhler curve $K$ obtained from them is defined as the points $i/n$, $i=0,1,\dots,n$ by $K(0)=0$ and $K(i/n)=s_i/s_n$, where $s_i=\sum_{j=1}^{i} x_j$. Therefore, $K(i/n)$ represents the fraction of the total $s_n$ contributed by the $i$th largest numbers. Thus, by the fraction $i/n$ of the largest numbers, the Leimkuhler curve is defined for all $u\in[0,1]$ by linear interpolation between adjacent point $(i/n,K(i/n))$, thus producing a polygonal or broken-line graph in rectangular coordinates from the origin to the end $(1,1)$.

The following formal definition of a Leimkuhler curve is due to \cite{sarabia_2008a}.
\begin{definition}\normalfont Let $X\in{\cal X}$ a non-negative random variable with cumulative distribution function $F_X(x)=\Pr(X\leq x)$ and inverse distribution function $F_X^{-1}(z)=\inf\{x: F_X(x)\geq z\}$. The Leimkuhler curve $K_X(u)$, $0\leq u\leq 1$, corresponding to $X$ is defined by
\begin{eqnarray}
K_X(u)=\frac{1}{\mu_X}\int_{1-u}^1 F_X^{-1}(y)\,dy,\quad 0\leq u\leq 1,\label{lc}
\end{eqnarray}
where $\mu_X=\int_{\cal X} x\,dF_X(x)<\infty$, the mathematical expectation of $X$, is assumed to be finite.
\end{definition}

{The definition above allows for the continuous and discrete case. Suppose that $K_X(u)$ {given in \eqref{lc}} has first and second derivatives and that $F_X$ has a continuous nonzero first derivative. Then $K_X(u)$ satisfies the following properties \citetext{see for instance \citealp{sarabia_2008a}}:
\begin{enumerate}[$(i)$]
\item $K_X(0)=0$, $K_X(1)=1$.
\item $K_X(u)$ is a non-decreasing function.
\item $K_X(u)$ is a concave function.
\end{enumerate}}

Other alternatives definitions of a Leimkuhler curve can be seen in \cite{sarabia_2008a}. Furthermore, taking into account that Leimkuhler and Lorenz curves (say $L_X(u)$) are related via
\begin{eqnarray}
K_X(u)=1-L_X(1-u)\label{relation}
\end{eqnarray}
the proof of the Theorem above can be carried out by taking into account the characterization of a Lorenz curve given, among others, in Theorem 1 in \cite{sarabiaetal_2010} and Theorem 2 in \cite{sordoetal_2014}.



The counterpart of the Gini index \citetext{see \citealp{burrell_1991, burrell1992}} defined for Lorenz curves to the Leimkuhler curve is given by
\begin{eqnarray}
G(\theta)=2\int_{0}^{1}K_X(u;\theta)\,du-1.
\label{gib}
\end{eqnarray}

This index takes values in the unit interval and measures the subject category's egalitarianism. The lower the index, the less egalitarian the subject category; the higher the index, the more egalitarian the subject category.

The authors \cite{kakwani1980} and \cite{yitzhaki_1983} proposed a generalization of the Gini index that applies to Lorenz curves. Their proposed measure, which has been adapted for the Leimkuhler curve, is known as the generalized Gini index. It is defined by the equation

\begin{eqnarray}
G_r(\theta)=r(r+1)\int_0^1(1-u)^{r-1}K_X(u;\theta)\,du-1,\quad r >0.\label{yi}
\end{eqnarray}

The parameter $r$ plays a crucial role in determining the sensitivity of the generalized Gini index to changes in the income distribution. Specifically, when $r\to 0$, changes at the upper end of the distribution have the only impact on the value of the index and this index does not distinguish among Leimkuhler curves. Conversely, as $r$ increases towards infinity, the index becomes increasingly sensitive to changes at the lower end of the distribution. In the limit, the value of the index is determined exclusively by the lowest income, reflecting the notion that the overall social welfare is contingent solely on the well-being of the poorest member of society. It is worth noting that the standard Gini index, as represented by equation \eqref{gib}, is retrieved when $r$ is set to 1.


On the other hand, the Pietra index \citetext{\citealp{eliazarandsokolov_2010}}, defined as
\begin{eqnarray}
P(\theta)=\max_{0\leq u \leq 1}|K_X(u;\theta)-u|,\label{pi}
\end{eqnarray}
i.e. the distance between the Leimkuhler curve and the line of perfect equality.

{
Although, in practice, the citation distributions are integer-valued, despite the variable of interest in the paper being discrete, the study is limited to considering distributional modeling with continuous approximations.}

This work does not intend to keep the distribution of the number of citations. In this sense, classical discrete distributions, such as the Poisson and negative binomial distributions, are the appropriate modeling for this. We will follow the spirit of Leimkuhler's work \citetext{see \citealp{leimkuhler_1967}} in which the curve named after him will be used to adjust the fraction of documents, $x$, that are most productive, i.e., $0<x<1$, while $F_X(x)$ denotes the proportion of total productivity contained in that fraction. It will be shown that the Leimkuhler curve constructed from (1), regardless of the distribution function from which it is based and its domain or support is a valuable tool for the intended purpose. Later, this task can be improved by enriching the Leimkuhler curve on which it is based.

\begin{example}\normalfont Let $Z$ be a power distribution with cumulative distribution function given by $F_Z(z)=z^{1/\theta}$, $z\in(0,1)$, $\theta>0$ and mean $\mu_Z=(1+\theta)^{-1}$. Now, from \eqref{lc} we get the corresponding Leimkuhler curve given by {\citetext{this model has been considered by \citealp{sarabia_2008b}}}
\begin{eqnarray}
K_X(u;\theta)=1-\bar u^{1+\theta},\label{plc}
\end{eqnarray}
where $\bar u=1-u$, $u\in[0,1]$. {A simple generalization of this Leimkuhler curve can be obtained from the Lorenz curve given in \cite{sarabiaetal_1999} by using \eqref{relation}  (see expression (15) in that paper) and is given by
\begin{eqnarray}
K_X(u;\theta,k)=1-(1-u^k)\bar u^{\theta},\label{gplc}
\end{eqnarray}
where $0<k\leq 1$. Note that when $k=1$ expression \eqref{gplc} reduces to \eqref{plc} as a special case. The corresponding Gini indices to these Leimkuhler curves are given by
\begin{eqnarray}
G(\theta) &=& \frac{\theta}{2+\theta},\label{gip}\\
G(\theta,k) &=& \frac{2\Gamma(k+1)\Gamma(\theta+1)}{\Gamma(\theta+k+2)}+\frac{\theta-1}{\theta+1},\label{gip1}
\end{eqnarray}
for \eqref{plc} and \eqref{gplc}, respectively.}

Now, from \eqref{yi} we get
\begin{eqnarray*}
G_r(\theta)=\frac{r\theta}{1+r+\theta},
\end{eqnarray*}
{for the Leimkuhler curve given in \eqref{plc}. An expression of this index, which is not reproduced here, can be obtained for \eqref{gplc} in terms of the incomplete beta function.}

On the other hand, the Pietra index given in \eqref{pi}, {for \eqref{plc},} is reached at $u=1-(1+\theta)^{-1/\theta}$ and is given by \begin{eqnarray}
P(\theta)=\theta(1+\theta)^{-1/\theta-1}.\label{pip}
\end{eqnarray}

{Numerical computation is needed to get the Pietra index for the Leimkuhler curve given in \eqref{gplc}. These two Leimkuhler curves will be denoted as the power Leimkuhler curve and the generalized power Leimkuher curves, respectively}.$\Box$
\end{example}

{
\begin{example}\normalfont Let $Z$ be a classical Pareto distribution with cumulative distribution function given by $F_Z(z)=1-\left(z/\sigma\right)^{-1/\theta}$, $z\geq \sigma>0$, $0<\theta<1$ and mean $\mu_Z=\sigma(1-\theta)^{-1}$. Now, from \eqref{lc} we get the corresponding Leimkuhler curve given by {\citetext{this model has been considered by \citealp{sarabia_2008a}}}
\begin{eqnarray}
K_X(u;\theta)=u^{1-\theta},\label{plcp}
\end{eqnarray}

The corresponding Gini index to this Leimkuhler curve is given by
\begin{eqnarray}
G(\theta)=\frac{\theta}{2-\theta}.\label{pgip}
\end{eqnarray}

Now, from \eqref{yi} we get
\begin{eqnarray*}
G_r(\theta)=\frac{\Gamma(2+r)\Gamma(2-\theta)}{\Gamma(2+r-\theta)}-1.
\end{eqnarray*}

On the other hand, the Pietra index given in \eqref{pi} is reached at $u=(1+\theta)^{1/\theta}$ and is given by \begin{eqnarray}
P(\theta)=\theta(1+\theta)^{1/\theta-1}.\label{pip1}
\end{eqnarray}

In advance this Leimkuhler curve will be denoted as the Pareto Leimkuhler curve.$\Box$
\end{example}
}

We have considered two data sources to see how these two models work. They correspond to two JCR subject categories in the WOS database: Operations Research \& Management Science (OR \& MS), and Statistic \& Probability (Stat. \& Prob.), and some descriptive statistics of them are shown in Table \ref{tab1}. The total number of "research articles" in 2016{, referred to these subject categories (OR \& MS and Stat. \& Prob.) was 12,767}. In addition, we identified all citations received within a 5-year citation window (2016-2020).

\begin{table}[htbp]
\caption{Descriptive statistics of the {variable number of citations}\label{tab1}}
\begin{center}
\begin{tabular}{lccc}\hline
& \multicolumn{1}{c}{\bf Stat. \& Prob.} && \multicolumn{1}{c}{\bf OR \& MS}\\ \hline
& cites  && cites \\ \hline
{$\min$} & 0 && 0\\
{$\max$} & 908 && 84\\
Mean & 13.81  && 14.08 \\
Variance & 273.90  &&  50.41 \\
ID & 19.83  &&  3.58 \\
{total citations} & \multicolumn{1}{c}{94267} && \multicolumn{1}{c}{83649}\\
$n$ & 6826 && 5941 \\ \hline
\end{tabular}
\end{center}
\end{table}
To estimate the parameters in equations \eqref{plc} and \eqref{plcp}, we employed a non-linear least squares fit. To derive our results, we adopted a methodology that entailed minimizing the sum of squared discrepancies between the empirical data points of the Leimkuhler curve and the corresponding values predicted by the theoretical models. More specifically, if we denote a Leimkuhler curve as $K_X(u;\theta)$, the objective of the analysis can be expressed mathematically as a minimization problem:
\begin{eqnarray*}
\min_{\theta}\sum_{i=1}^{n}\left(\frac{s_i}{s_n}-K_X\left(\frac{i}{n};\theta\right)\right)^2,
\end{eqnarray*}
that is, the sum of the squared discrepancies between the empricial data points of the Leimkuhler curve, $K(u_i)=s_i/s_n$, $u_i=i/n$ and the corresponding values predicted by the theoretical model, $K_X(u_i;\theta)$. Here $\theta$ represents the parameter vector to be estimated. The non-linear least squares method allows for the estimation of non-linear parameters, such as those in equations {\eqref{plc}, \eqref{gplc} and \eqref{plcp}} by iteratively minimizing the sum of the squared differences between the theoretical and empirical curves until convergence is reached.

Table \ref{tab2} shows the estimation results for the { three models considered given in equations \eqref{plc}, \eqref{gplc} and \eqref{plcp}}, respectively. To estimate the standard errors of the coefficients, we used a consistent estimate of the covariance matrix, as described in \cite{amemiya_1985}. We also provide the consistent Akaike's information criteria (CAIC), proposed by
\cite{bozdogan_1987}. The CAIC is expressed as $k(1+\log n)-2\log\ell$. This last measure of model selection was chosen to overcome the tendency of the AIC to
overestimate the complexity of the underlying model since it lacks
certain properties of asymptotic consistency as it does not
directly depend on the sample size. Then, to calculate
the CAIC, a correction factor based on the sample size is used to
compensate for the overestimating nature of AIC. Note that a model with a lower CAIC value is preferred to one with a higher value. To assess the goodness of fit among the models, various measures were employed, including the mean squared error (MSE), the maximum absolute error (MAX), and the mean absolute error (MAE). These measures are formulated as follows:
\begin{eqnarray*}
\mbox{MSE} &=& \frac{1}{n}\sum_{i=1}^{n}\left[\frac{s_i}{s_n}-K_X(u_i;\widehat\theta)\right]^2,\\
\mbox{MAX} &=& \max_{i=1,\dots,n}\left|\frac{s_i}{s_n}-K_X(u_i;\widehat\theta)\right|,\\
\mbox{MAE} &=& \frac{1}{n}\sum_{i=1}^{n}\left|\frac{s_i}{s_n}-K_X(u_i;\widehat\theta)\right|.
\end{eqnarray*}

We also present the Gini and the Pietra indices obtained by using expressions provided in {\eqref{gip}, \eqref{gip1}, \eqref{pgip}, \eqref{pip} and \eqref{pip1}}. We used a code implemented in Wolfram Mathematica (v.13.0) to compute the estimated parameters, as well as the MSE, MAX, MAE, Gini, and Pietra indices. The code is straightforward and easy to use. It is worth noting that a plot of the CDF of the residuals against the CDF of a normal distribution has been represented in all the cases considered, not resulting in a straight line. Therefore, the residuals appear to violate the normality hypothesis. Hence, maximum likelihood estimation seems more appropriate.

In Figure \ref{fig1} the cumulative number of citations, $K_X(u)$, is plotted against the cumulative proportion, $u$, of articles and compared with the estimated Leimkuhler curves given in {\eqref{plc}, \eqref{gplc} and \eqref{plcp} for the sources of data, respectively.} The graphs show that standard models accurately represent the empirical Leimkuhler curves.
\begin{figure}[htbp]
\begin{center}
\includegraphics[page=1,width=0.45\linewidth]{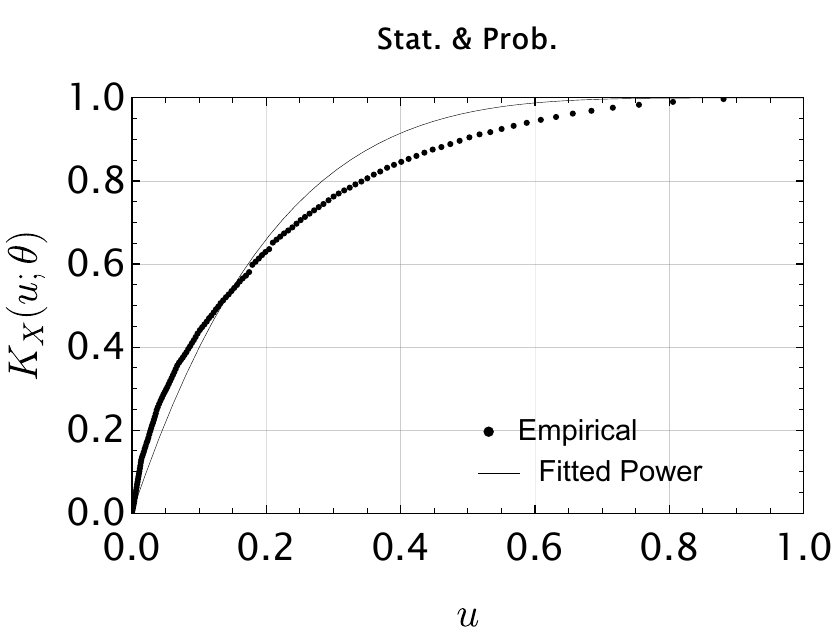}
\includegraphics[page=9,width=0.45\linewidth]{figures.pdf}
\includegraphics[page=2,width=0.45\linewidth]{figures.pdf}
\includegraphics[page=10,width=0.45\linewidth]{figures.pdf}
\includegraphics[page=3,width=0.45\linewidth]{figures.pdf}
\includegraphics[page=11,width=0.45\linewidth]{figures.pdf}
\caption{Empirical and fitted, by non-linear least square, Leimkuhler curves based on the data considered.
{From top to down they correspond to equations given in \eqref{plc}, \eqref{gplc} and \eqref{plcp}}, respectively.}  \label{fig1}
\end{center}
\end{figure}

\section{Beyond the standard model\label{s3}}
Empirical evidence suggests that a simple model of this nature may be too inflexible and insufficient in explaining the phenomenon under study. It is more appropriate to assume that the phenomenon is heterogeneous and that a mixture model may better capture the underlying complexity. Furthermore, the contagious and other missing characteristics of the data suggest that the pattern of citation behavior varies intrinsically across different journals in the subject category.

In practice, over-dispersion, i.e. the variance is greater than the mean, arises due to the heterogeneity of the population, meaning that individuals have constant but unequal probabilities of experiencing an event. In this case, each member of the population has a value of the parameter $\theta\in\Theta$, which follows a probability density function $g(\theta)$. Therefore, $\theta$ can be considered an inhomogeneity parameter that varies from journal to journal, following a probability density function $g(\theta;\Upsilon)$, and the actual Leimkuhler curve can be constructed in this way. Here, $\Upsilon$ represents the hyperparameters of the mixing distribution.

The use of mixture models in informetric and bibliometrics is familiar and has been used to get long-tail citation distributions. See for example \cite{burrellandcane_1982} and \cite{burrell_2005}, among others. It is widely recognized that informetric data poses some specific challenges because they tend to exhibit a heavily skewed distribution.

In short, although the basic model may be exact within a group of journals, the parameter $\theta$ frequently will vary randomly across groups of journals because of random variation in citations and articles considered, for example. Both of these data sources exhibit empirical over-dispersion, meaning that their variance is greater than their mean. Table \ref{tab1} displays the index of dispersion ($\mbox{ID}=var(X)/\E(X)$), as well as the mean and variance of the empirical data for both the citations and articles variables across the two subject categories of journals. Furthermore, there is evidence to suggest that the over-dispersion may be linked to the heterogeneity of the population across subject categories. In such cases, the parameter $\theta$ can be treated as a random variable that assumes different values across various journals within the subject category, thus capturing the inherent uncertainty associated with this parameter, and varying from individual to individual based on a probability density function.

Then, we propose the following Leimkuhler curve instead of the one provided in \eqref{lc}.
\begin{proposition}\normalfont If $K_X(u;\theta)$ is a Leimkuhler curve obtained from \eqref{lc}, then
\begin{eqnarray}
K_X(u;\Upsilon)=\E_{\Theta}[K_X(u;\theta)]=\int_{\Theta}K_X(u;\theta)g(\theta;\Upsilon)\,d\theta\label{mlc}
\end{eqnarray}
is also a proper Leimkuhler curve.
\end{proposition}
\begin{proof} The result is direct.
\end{proof}

The Gini index in this case can be computed also by mixture as we show in the next result.
\begin{proposition}\normalfont If $K_X(u;\Upsilon)$ is a proper Leimkuhler curve obtained from \eqref{mlc}, then the Gini index is given by
\begin{eqnarray}
G(\Upsilon)=\E_{\Theta}[G(\theta)]=\int_{\Theta}G(\theta)g(\theta;\Upsilon)\,d\theta.\label{gmlc}
\end{eqnarray}
\end{proposition}
\begin{proof} It is simple.
\end{proof}

\subsection{Specific examples}
{Note that the Power Leimkuhler curve provided in \eqref{gplc} can be rewritten as
\begin{eqnarray}
K_X(u;\theta)=1-(1-u^k)\exp(\theta \log\bar u),\quad 0\leq u < 1.\label{plc1}
\end{eqnarray}}

Thus, attending to \eqref{plc1} any probability
density function with a closed form for its moment generating function, such as the gamma,
the inverse Gaussian, the generalized inverse Gaussian, and Rayleigh distributions should be good candidates
to introduce here.

\subsubsection{The gamma case}
Assume first that the parameter $\theta\in\Theta$ follows a gamma distribution with shaper parameter $\alpha>0$ and rate parameter $\beta>0$, i.e. $g(\theta;\Upsilon) \propto \theta^{\alpha-1}\exp(-\beta \theta)$ where $\theta>0$ and $\Upsilon=(\alpha,\beta)$. It is well-known that its moment generating function is given by
\begin{eqnarray}
M_{\Theta}(t;\Upsilon)=\E[\exp(t\Theta)]=\left(\frac{\beta}{\beta-t}\right)^{\alpha}.\label{mgfg}
\end{eqnarray}

Now, using \eqref{plc}, \eqref{mlc} and \eqref{mgfg} it is straightforward to get the following mixture Leimkuhler curve
{
\begin{eqnarray}
K_X(u;\Upsilon) &=& 1-\bar u \psi_G(u;\Upsilon),\quad 0\leq u < 1.\label{mlcg}\\
K_X(u;\Upsilon,k) &=& 1-(1-u^k)\psi_G(u;\Upsilon),\quad 0\leq u < 1,\label{mlcg1}
\end{eqnarray}
where
\begin{eqnarray*}
\psi_G(u;\Upsilon) = \left(\frac{\beta}{\beta-\log\bar u}\right)^{\alpha}.
\end{eqnarray*}
}

The special case for which $\alpha=1$ derives in the mixture Leimkuhler curve obtained from the exponential mixing distribution.

From \eqref{gmlc} we get that the mixture Gini index results
\begin{eqnarray*}
G(\Upsilon)=\alpha (2\beta)^{\alpha} \Gamma(-\alpha,2\beta)\exp(2\beta),
\end{eqnarray*}
where $\Gamma(a,b)=\int_b^{\infty}t^{a-1}\exp(-t)\,dt$ is the incomplete gamma function.

\begin{eqnarray*}
G_r(\Upsilon)=\alpha r ((1+r)\beta)^{\alpha} \Gamma(-\alpha,(1+r)\beta)\exp[(1+r)\beta].
\end{eqnarray*}

In advance \eqref{mlcg} will be referred as the PG, denoting that is the mixture of a power Leimkuhler curve with the gamma distribution {and \eqref{mlcg1} as GPG}.

\subsubsection{The inverse Gaussian case}
Consider now that $\Theta$ follows an inverse Gaussian distribution with probability density function \citetext{see \citealp{folksandchhikara_1978}; and \citealp{seshadri_1983}} given by
\begin{eqnarray*}
g(\theta;\Upsilon)=\sqrt{\frac{\beta}{2\pi\theta^3}}\exp\left[
-\frac{\beta}{2\alpha^2\theta}(\theta-\alpha)^2\right],\quad \theta>0,
\end{eqnarray*}
where $\alpha>0$ is the mean of the distribution and $\beta>0$. Its variance is given by $\alpha^3/\beta$ and the moment generating function is given by
\begin{eqnarray}
M_{\Theta}(t;\Upsilon)=\E[\exp(t\Theta)]=\exp\left[\frac{\beta}{\alpha}\left(1-\sqrt{1-2\alpha^2 t/\beta}\right)\right].\label{mgfig}
\end{eqnarray}

{
Using now \eqref{plc}, \eqref{mlc} and \eqref{mgfig} we have the mixture Leimkuhler curve given by
\begin{eqnarray}
K_X(u;\Upsilon)=1-\bar u \exp[\psi_{IG}(u;\Upsilon)],\quad 0\leq u < 1,\label{mlcig}\\
K_X(u;\Upsilon,k)=1-(1-u^k) \exp[\psi_{IG}(u;\Upsilon)],\quad 0\leq u < 1,\label{mlcig1}
\end{eqnarray}
where
\begin{eqnarray*}
\psi_{IG}(u;\Upsilon) = \frac{\beta}{\alpha}\left(1-\sqrt{1-2\alpha^2 \log\bar u/\beta}\right).
\end{eqnarray*}
}

We will use the term PIG to denote the Leimkuhler curve given in \eqref{mlcig} because it is the power Leimkuhler curve with the inverse Gaussian distribution {and GPIG for the Leimkuhler curve given in \eqref{mlcig1}}. In this case, numerical computation is required to get the Gini index since non-closed-form expressions exist for \eqref{gmlc}.

{\subsection{The Pareto-confluent hypergeometric distribution case}
The confluent hypergeometric distribution due to \cite{gordy_1998} has a probability density function given by
\begin{eqnarray}
g(\theta;\Upsilon)=\frac{\theta^{\alpha-1}(1-\theta)^{\beta-1}\exp(-k\theta)}{B(\alpha,\beta)\,_1F_1(\alpha;\alpha+\beta;-k)},\quad
0<\theta<1,\label{chd}
\end{eqnarray}
where $\Upsilon=(\alpha,\beta,k)$ with $\alpha>0$, $\beta>0$, $-\infty<k<\infty$ and $_1F_1(\cdot;\cdot;\cdot)$ is the confluent hypergeometric function \citetext{see \citealp[p. 505]{abramowitzandstegun_1968}}, given by
\begin{eqnarray*}
_1F_1(a;b;z)=\frac{\Gamma(b)}{\Gamma(b-a)\Gamma(a)}\int_0^1 t^{a-1}(1-t)^{b-a-1}\exp(z t)\,dt.
\end{eqnarray*}

If $k=0$, then this density reduces to the standard beta density. Now, by using \eqref{plcp}, \eqref{mlc} and \eqref{chd} we get the mixture Leimkuhler curve given by
\begin{eqnarray}
K_X(u;\Upsilon) &=& \int_0^1 u\exp(-\theta\log u)\frac{\theta^{\alpha-1}(1-\theta)^{\beta-1}\exp(-k\theta)}{B(\alpha,\beta)\,_1F_1(\alpha;\alpha+\beta;-k)}\,d\theta
\nonumber\\
&=&
\frac{u\,_1F_1(\alpha;\alpha+\beta;-k-\log u)}{_1F_1(\alpha;\alpha+\beta;-k)}=
\frac{_1F_1(\beta;\alpha+\beta;k+\log u)}{_1F_1(\beta;\alpha+\beta;k)},\;\;0\leq u\leq 1,\label{mlcc}
\end{eqnarray}
where we have used the Kummer transformation \citetext{see \citealp[p. 505]{abramowitzandstegun_1968}} $_1F_1(a;b;-c)=\exp(-c)\,_1F_1(b-a;b;c)$.

In advance \eqref{mlcc} will be referred as the PaGB, denoting that is the mixture of a Pareto Leimkuhler curve with the confluent hypergeometric distribution (generalization of the standard beta distribution).

When $k=0$ the Leimkuhler curve given in \eqref{mlcc} reduces to
\begin{eqnarray*}
K_X(u;\Upsilon) =\,
_1F_1(\beta;\alpha+\beta;\log u),\quad 0\leq u\leq 1,
\end{eqnarray*}
which corresponds to the mixture of the Leimkuhler curve given in \eqref{plcp} with the standard beta density.

}

\subsection{{Stochastic ordering}}
We are examining the Leimkuhler ordering in this study. The Leimkuhler curve can establish an unambiguous ordering when two distribution functions have non-intersecting Leimkuhler curves \citetext{see \citealp{Atkinson1970}, \citealp{Dasgupta1973}, and \citealp{Shorrocks1983}, among others}. The class $\cal L$ is defined as all non-negative random variables with a positive finite expectation, and the Leimkuhler partial order $\leq_L$ is defined on the class $\cal L$ by
\begin{eqnarray*}
X\leq_L Y \Longleftrightarrow F_X(x)\leq F_Y(y) \Longleftrightarrow K_X(u)\leq K_Y(u), \forall u\in[0,1].
\end{eqnarray*}

If $X\leq_L Y$, then $X$ exhibits less inequality than $Y$ in the
Leimkuhler sense. The next result shows that the family of Leimkuhler curves given in \eqref{mlcg} and \eqref{mlcig}
are ordered with respect to the $\Upsilon$ vector of parameters.
\begin{proposition}
{Let $\Upsilon_1=(\alpha_1,\beta_1,k)$, $\Upsilon_2=(\alpha_2,\beta_2,k)$, then the Leimkuhler curves given in \eqref{mlcg} and \eqref{mlcg1} are ordered with respect to $\alpha$ and $\beta$, i.e.
\begin{enumerate}[$(i)$]
\item if $\alpha_1<\alpha_2$ we have that $K_X(u;\Upsilon_1)\geq K_X(u;\Upsilon_2)$, $\forall \beta>0$, $\forall k\in (0,1]$.
\item if $\beta_1<\beta_2$ we have that $K_X(u;\Upsilon_1)\leq K_X(u;\Upsilon_2)$, $\forall \alpha>0$, $\forall k\in (0,1]$.
\end{enumerate}}
\end{proposition}
\begin{proof} By computing the derivative of {\eqref{mlcg} or \eqref{mlcg1}} with respect to $\alpha$ and $\beta$ we get after some simple computations that
\begin{eqnarray*}
\frac{\partial K_X(u;\Upsilon)}{\partial \alpha} &=& -\left[1-K_X(u;\Upsilon)\right]\log\left(\frac{\beta}{\beta-\log\bar u}\right)>0,\\
\frac{\partial K_X(u;\Upsilon)}{\partial \beta} &=& {\frac{\left[1-K_X(u;\Upsilon) \right]\alpha \log \bar u}{\beta(\beta-\log\bar u)}<0,}
\end{eqnarray*}
respectively. Hence the result.
\end{proof}
{
\begin{proposition}
Let $\Upsilon_1=(\alpha_1,\beta_1,k)$, $\Upsilon_2=(\alpha_2,\beta_2,k)$, then the Leimkuhler curves given in \eqref{mlcig} and \eqref{mlcig1} are ordered with respect to $\alpha$ and $\beta$, i.e.
\begin{enumerate}[$(i)$]
\item if $\alpha_1<\alpha_2$ we have that $K_X(u;\Upsilon_1)\geq K_X(u;\Upsilon_2)$, $\forall \beta>0$, $\forall k\in (0,1]$.
\item if $\beta_1<\beta_2$ we have that $K_X(u;\Upsilon_1)\geq K_X(u;\Upsilon_2)$, $\forall \alpha>0$, $\forall k\in (0,1]$.
\end{enumerate}
\end{proposition}
\begin{proof} Again, some computations provide the following derivative of {\eqref{mlcig} or \eqref{mlcig}} with respect to $\alpha$ and $\beta$,
\begin{eqnarray*}
\frac{\partial K_X(u;\Upsilon)}{\partial \alpha} &=& -\left[1-K_X(u;\Upsilon)\right]\frac{\beta}{\alpha^2}\left[-\left(1-\sqrt{1-2\alpha^2 \log\bar u/\beta}\right)+\frac{2\alpha^2\log\bar u/\beta}{\sqrt{1-2\alpha^2 \log\bar u/\beta}}\right]>0,\\
\frac{\partial K_X(u;\Upsilon)}{\partial \beta} &=& -\frac{\left[1-K_X(u;\Upsilon)\right]}{\alpha}\left(1-\sqrt{1-2\alpha^2 \log\bar u/\beta}-\frac{\alpha^2\log\bar u/\beta}{\sqrt{1-2\alpha^2 \log\bar u/\beta}}\right)>0,
\end{eqnarray*}
respectively. Thus, the result follows.
\end{proof}}

{
\begin{proposition}
Let $\Upsilon_1=(\alpha,\beta,k_1)$, $\Upsilon_2=(\alpha,\beta,k_2)$, then the {Leimkuhler curves given in \eqref{mlcg1} and \eqref{mlcig1} are} ordered with respect to $k$, i.e.
\begin{enumerate}[$(i)$]
\item if $k_1<k_2$ we have that $K_X(u;\Upsilon_1)\leq K_X(u;\Upsilon_2)$, $\forall \alpha,\beta>0$.
\end{enumerate}
\end{proposition}
\begin{proof} Again, some computations provide the following derivative of {\eqref{mlcg1} or \eqref{mlcig1}} with respect to $k$,
\begin{eqnarray*}
\frac{\partial K_X(u;\Upsilon)}{\partial k} = \frac{\left[1-K_X(u;\Upsilon_1)\right]k u^{k-1}}{1-u^k}>0.
\end{eqnarray*}

Hence, the result follows.
\end{proof}
}

\section{Calculations\label{s4}}

{Table \ref{tab2} displays the estimation results for the models considered here. The computational time needed for obtaining parameter estimates and model comparison measures is minimal compared to standard models.}

As evidenced by the tables, the novel models exhibit superior performance with in general smaller values of MSE, MAX and MAE. Furthermore, Gini and Pietra indices are nearer to the empirical values in comparison with the initial Leimkuhler curves (0.6112 and 0.4625 for Stat. \& Prob. data, and 0.5166 and 0.3803 for OR \& MS data, respectively). Notably, the new proposed functional forms yield the best fit in this comparison, with the exception of the Pareto and Pareto mixture model which work worst than the rest of the models.

\begin{sidewaystable}
\caption{Estimated parameters (standard error between parenthesis), using non-linear least square, for the different Leimkuhler curves considered and some statistics.\label{tab2}}

\begin{center}

\hspace*{-1cm}
\begin{tabular}{llrrrrrrrrrr}\hline
& Model &  \multicolumn{4}{c}{Estimated parameters} & MSE & MAX & MAE & CAIC &  Gini & Pietra\\
\cline{3-6}
&&&&&&&&&\\
 & & $\widehat\theta$ & $\widehat k$ &$\widehat\alpha$ & $\widehat\beta$  & & & &  \\ \hline

\hline

{\bf Stat. \& Prob.} &
  Power  & 3.832 & & & & 0.0027 & 0.0835 & 0.0471 &  -20778.40 &   0.6571  & 0.5276 \\
 & &  (0.0131) &&&&&&&\\
 & PG & &  & 0.701 & 0.102 & 0.0003 & 0.0318 & 0.0149  &  -36227.40 &   0.5910  & 0.4574\\
 &   & & & (0.0275) & (0.0058) &&&&&\\

 & PIG & &  & 9.305 & 2.227 & $7.7\cdot 10^{-5}$ & 0.0169 & 0.0076  &  -45277.80 &   0.6011  & 0.4536\\
 &   & & & (0.0275) & (0.0058) &&&&&\\

 & {   GP}  & 1.859 & 0.514 & & &$1.3\cdot 10^{-4}$ & 0.0319 & 0.0089 & -41525.70  &  0.6204 & 0.4636\\
 & &  (0.0042) & (0.0005) &&&&&&\\
 & {GPG} & & 0.554 & 1.514 & 0.596 & $6.7\cdot 10^{-5}$ & 0.0236 & 0.0060 &  -46149.90  &  0.6071 & 0.4608\\
 & & & (0.0007)   & (0.0184) & (0.0096) &&&&&\\

 & {GPIG} & & 0.571 & 2.919 & 2.416 & $5.9\cdot 10^{-5}$ & 0.0209 & 0.0061 & -46983.00 & 0.6070 & 0.4591\\
 & & & (0.0010) &  (0.0182) & (0.0230) & & & & & \\
& {Pa} & 0.645 &  &  &  & 0.0067 & 0.1244 & 0.0730 & -14744.30 & 0.4756 & 0.3373\\
 & & (0.0011) &&  &  & & & & & \\
 & {PaGB} & & -25.483 & 11.042 & 14.080 & 0.0078 & 0.4231 & 0.0797 & -15703.10 & 0.4984 & 0.3840\\
 & & & (0.0360) & (0.1543) & (0.1381) & & & & & \\ \hline
{\bf OR \& MS} &
  Power  & 2.767 & & & & 0.0044 & 0.0968 & 0.0601 &  -15361.20 &   0.5804 & 0.4548\\
 & & (0.0140) &&&&&&&\\
 & PG &  & & 0.392 & 0.055 & $1.5\cdot 10^{-5}$ & 0.0261 & 0.0108 &  -35415.90  &  0.5021 & 0.3838\\
 & &  & & (0.0013) & (0.0003) &&&&&\\

 & PIG &  & & 14.035 & 1.029 & $9.2\cdot 10^{-6}$ & 0.0096 & 0.0024 &  -52011.20  &  0.5188 & 0.3812\\
 & & & & (0.033) & (0.0008) & & & & & \\

& {   GP}  & 0.928 & 0.492 & & & 0.0002 & 0.0324 & 0.0110 & -34430.30  & 0.5282 & 0.3803\\
 & & (0.004) & (0.0005) &&&&&&\\
 & {GPG} & & 0.558 & 0.445 & 0.212 & $4.1\cdot 10^{-5}$ & 0.0170 & 0.0052 & -43094.90 &  0.5115 & 0.3849\\
 & & & (0.0008) & (0.0028) & (0.0026) &&&&&\\

 & {GPIG} & & 0.799 & 10.765 & 0.742 & $7.5\cdot 10^{-6}$ & 0.0088 & 0.0019 & -53166.10 & 0.5165 & 0.3813\\
 & & & (0.0039) & (0.1181) & (0.0056) & & & & & \\

 & {Pa} & 0.606 &  &  &  & 0.0028 & 0.0783 & 0.0478 & -18027.20 & 0.4344 & 0.3305\\
 &  & (0.0009) &  &  &  & & & & & \\
 & {PaGB} & & -28.083 & 45.010 & 41.080 & 0.0029 & 0.473 & 0.0489 & -18506.80 & 0.4307 & 0.3276\\
 & & & (0.096) & (1.967) & (1.369) & & & & & \\
\hline

\multicolumn{10}{l}{\scriptsize PG: Mixture power-gamma Leimkuhler curve;  PIG: Mixture power-inverse Gaussian Leimkuhler curve}\\
\multicolumn{10}{l}{\scriptsize GPG: Mixture generalized power-gamma Leimkuhler curve; GPIG: Mixture generalized power-inverse Gaussian Leimkuhler curve}\\
\multicolumn{10}{l}{\scriptsize PaGB: Mixture Pareto-generalized beta Leimkuhler curve}
\end{tabular}
\end{center}
\end{sidewaystable}

Furthermore, we have provided a visual representation of the performance of each model in Figure \ref{fig2}. As expected, the mixture models fit the empirical Leimkuhler curve reasonably well and outperform the Power Leimkuhler models shown in Figure \ref{fig1}. In particular, the mixture based on the inverse Gaussian distribution is better.
\begin{figure}[htbp]
\begin{center}
\includegraphics[page=4,width=0.45\linewidth]{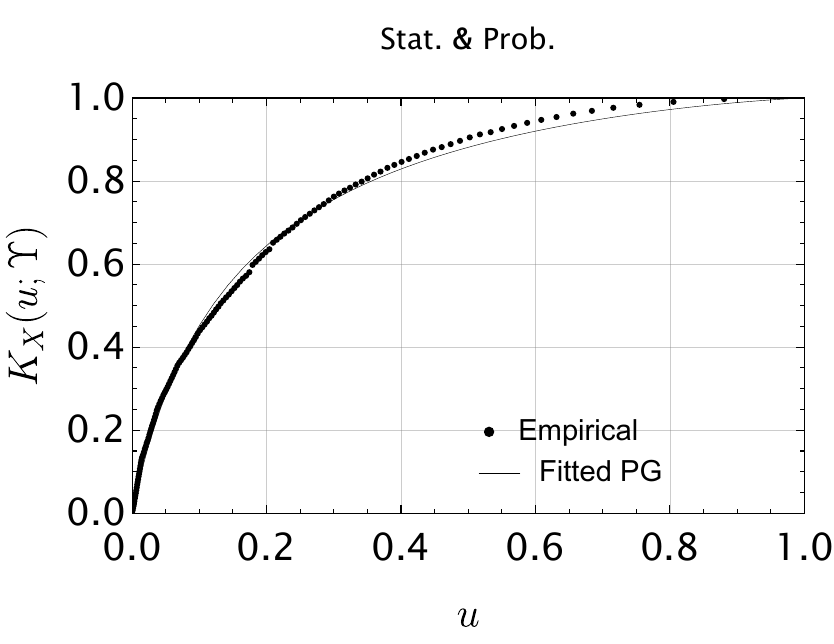}\hspace{0.25cm}
\includegraphics[page=12,width=0.45\linewidth]{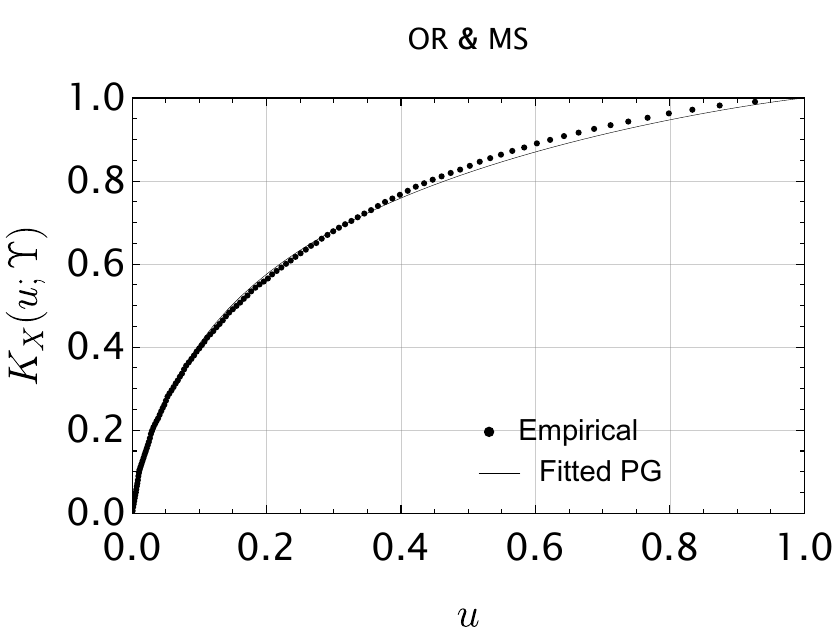}\\
\includegraphics[page=5,width=0.45\linewidth]{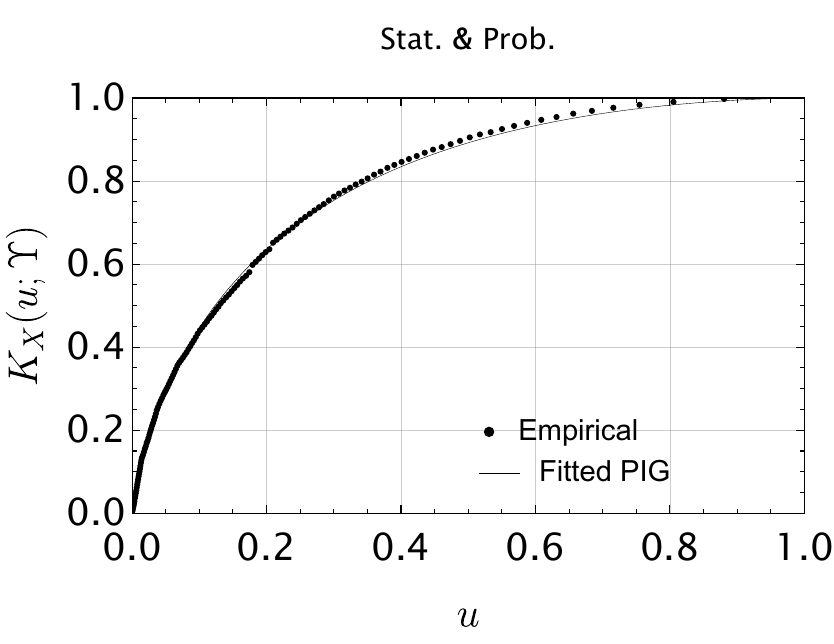}\hspace{0.25cm}
\includegraphics[page=13,width=0.45\linewidth]{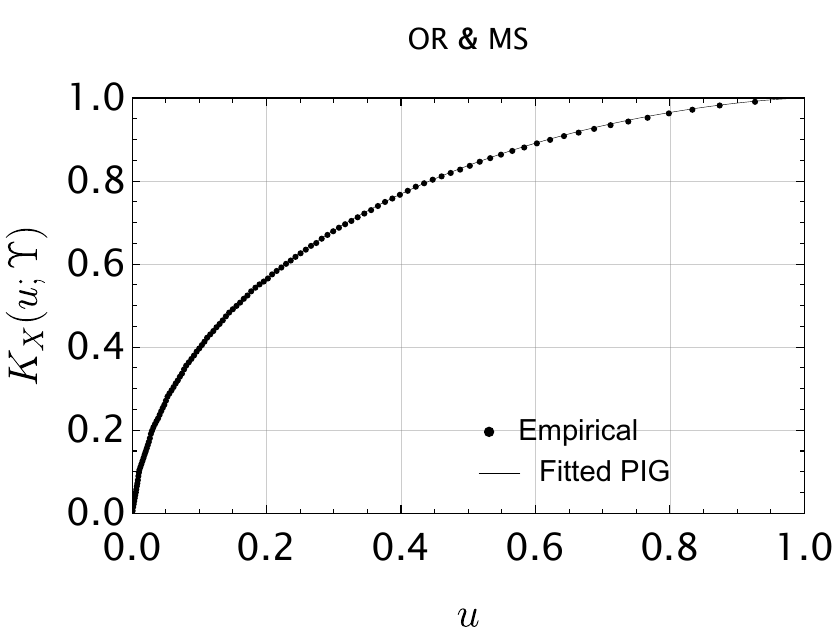}\\
\includegraphics[page=7,width=0.45\linewidth]{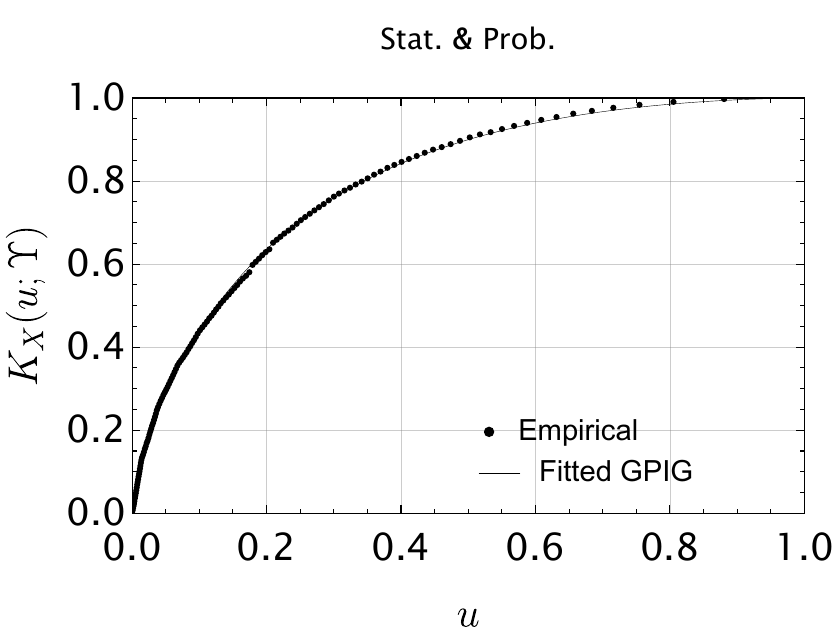}\hspace{0.25cm}
\includegraphics[page=15,width=0.45\linewidth]{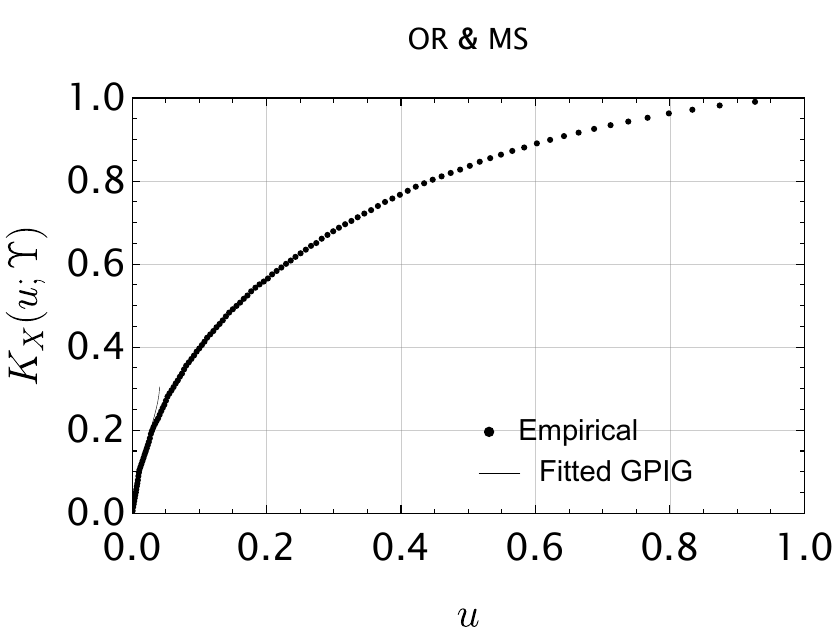}\\
\includegraphics[page=8,width=0.45\linewidth]{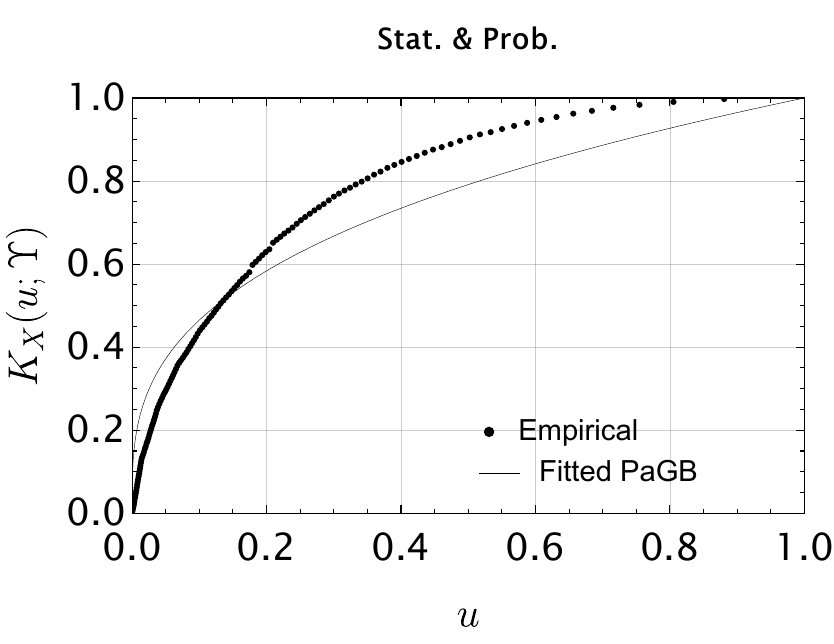}\hspace{0.25cm}
\includegraphics[page=16,width=0.45\linewidth]{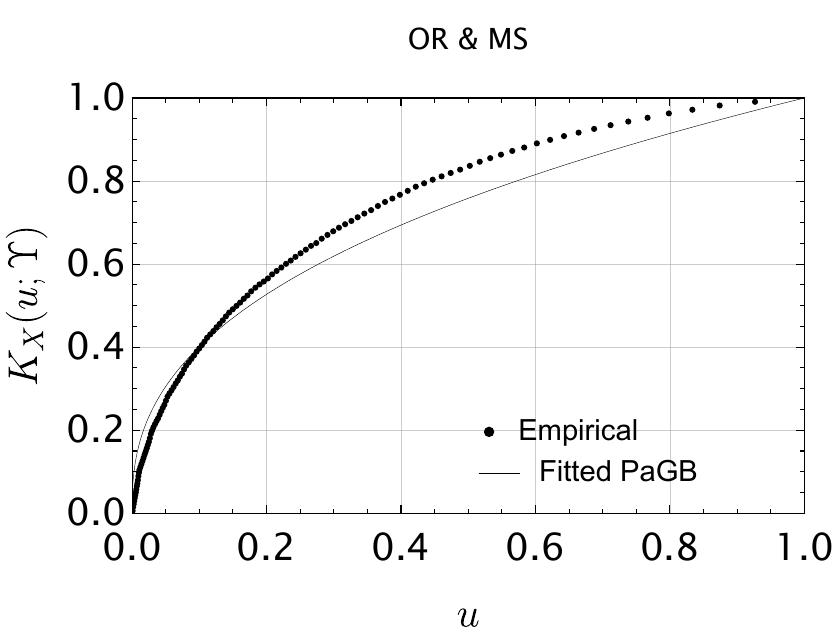}

\caption{Empirical and fitted, by non-linear least square, Leimkuhler curves based on the data considered and mixture models. The left and right sides correspond to the mixture power-gamma (PG) and mixture power-inverse Gaussian (PIG) Leimkuhler curves, given in equations \eqref{mlcg} and \eqref{mlcig}, respectively.}  \label{fig2}
\end{center}
\end{figure}

We found higher citation concentration indices, as measured by the Gini and Pietra indices, in the case of Stat. \& Prob. Several factors may explain these differences, as discussed below.

Citation frequency and volume: Research in OR \& MS often involves the application of mathematical models, optimization techniques, and decision frameworks to practical problems. This often leads to more frequent citations as practitioners and researchers refer to specific methods and solution approaches. In contrast, citations in the field of Stat. \& Prob. may be less frequent, especially for well-established theories and fundamental concepts. However, when new methodologies, algorithms, or data analysis techniques are introduced, citations may become more substantial \citetext{\citealp{adleretal_2009}; \citealp{jietal_2022}}.

Citation purposes: Citations in OR \& MS are often used to demonstrate the effectiveness and feasibility of particular optimization techniques, decision models, or algorithms. These citations often support the implementation of solutions in real-world scenarios. On the other hand, citations in the area of Stat. \& Prob. often serve to provide the theoretical foundation for data analysis methods, probability models, and statistical inference. They form the basis for rigorous methodologies and may be cited to rationalize statistical choices made in research (\citealp{adleretal_2009}).

Interdisciplinarity: While both fields have interdisciplinary connections, OR \& MS, with its focus on problem solving and decision making across sectors, may attract citations from diverse fields such as engineering, business, health care, and logistics. Similarly, citations within Stat. \& Prob. may have a stronger presence in disciplines related to data science, biology, economics, social sciences, and physical sciences, where statistical methods play a central role in analysis (\citealp{jietal_2022}).

Data sources: OR \& MS often relies on real-world data from operational and management contexts, such as supply chain or production data. In contrast, Stat. \& Prob. employs a wide range of data sources, including experimental data, surveys, and observations (\citealp{devore_2015}).

It's important to recognize that these distinctions are generalizations, and there may be variation among individual researchers and specific research topics within each subject category.

\section{Conclusions\label{s5}}

The Leimkuhler curve has become a widely used tool in informetrics and other fields for analyzing the concentration and inequality of resources and productivity. This paper proposed a new method for parameterizing the citation distribution using a mixture of Leimkuhler curves, which provides a more flexible and accurate framework for modeling citation distributions. We also considered measures of citation concentration across different fields, such as the Gini and Pietra indexes, to capture the differences in citation patterns across various research fields. Our results demonstrate the effectiveness of our proposed method and highlight the importance of considering measures of concentration and inequality in bibliometric analysis and other fields.

The distribution of citations can be characterized by a combination of Leimkuhler curves, with citation concentration measurements providing insights into how citations are concentrated among a select group of highly cited journals and papers. By employing this proposed method, a more comprehensive understanding of citation patterns can be achieved, leading to the effective identification of a core group of leading journals or papers. Empirical results demonstrate the method's efficacy in capturing differences in citation patterns across various fields, as evidenced by goodness of fit and predictive accuracy analyses. These findings hold significant implications for comprehending the dynamics of scientific knowledge productivity and dissemination, thereby making a valuable contribution to the field of bibliometrics.

As a limitation, we used two field examples to test the method. However, we consider that the main contribution of the paper lies in the theoretical and mathematical results. Specifically, in relation to the introduction of mixture distributions (called PG and PIG) to bibliometrics, and the obtaining of closed-form expressions for Leimkuhler curves.

We conducted an analysis of citations within the subject categories Operations Research \& Management Science (OR \& MS) and Statistics \& Probability (Stat. \& Prob.). Our calculations shed light on these categories in the following ways. We found higher citation concentration indices, as measured by the Gini and Pietra indices, in the case of Stat. \& Prob. Factors that may explain these differences include citation frequency and volume, citation purpose, interdisciplinarity, and data sources.

The results obtained from the use of Leimkuhler curves have significant applications in the field of library and information science. These curves provide valuable insight into trends and patterns in the development of collections and research areas, and allow for the visualization of the distribution of information sources over time (\citealp{gupta_2022}). When combined with bibliometric techniques, these curves can help identify influential authors or publications within a particular subject area, facilitate the building of comprehensive collections, and identify potential research collaborators or partners.

Leimkuhler curves serve as valuable tools for examining the evolution of a particular field or subject area over time. By comparing curves from different time periods, researchers can identify trends and changes in the distribution of information sources, providing crucial insights into the dynamic nature of research fields (\citealp{borgohain_2021}). Moreover, these curves are useful for evaluating the adequacy and balance of a library's collection over time. By analyzing the curves, librarians can pinpoint areas where adjustments in source acquisition may be necessary to ensure that the collection remains current and relevant. Furthermore, these curves hold relevance in the context of digital libraries, as they enable the study of usage patterns of electronic resources, which can inform collection development decisions and resource allocation strategies.

Finally, the use of gamma and inverse Gaussian distributions provides a lot of flexibility to the new models obtained, but beyond the underlying mathematical convenience, there is not theoretical justification for their use. This constitutes a weakness or limitation that must be underlined.

\end{document}